%% file: ZHS_SVD_bitalloc_MassMIMO_Spcom18.tex
\documentclass[conference]{IEEEtran}
\usepackage{algorithm}
\usepackage{algpseudocode}
\usepackage{tabu}
\usepackage{setspace}
\usepackage{makecell}
\usepackage{epsfig}
\usepackage[usenames, dvipsnames]{color}
\usepackage[roman]{parnotes}
\usepackage{placeins}

\usepackage{amsthm,amsmath,amssymb}
\usepackage{mathtools}
\usepackage{breqn}
\usepackage{amsfonts}
\usepackage{amssymb}
\usepackage{mathrsfs}
\usepackage{xparse}
\usepackage{multirow}
\ifCLASSINFOpdf
  \DeclareGraphicsExtensions{.pdf,.jpeg,.png}
\else
  \DeclareGraphicsExtensions{.eps}
\fi
\DeclareMathOperator{\tr}{tr}

\begin{document}
%
\title{Single-User mmWave Massive MIMO: SVD-based ADC Bit Allocation and Combiner Design}

\author{\IEEEauthorblockN{I. Zakir Ahmed  and Hamid Sadjadpour\\}
\IEEEauthorblockA{Department of Electrical Engineering\\
University of California, Santa Cruz\\
}
\and
\IEEEauthorblockN{Shahram Yousefi}
\IEEEauthorblockA{Department of Electrical and Computer Engineering\\
Queen's University, Canada\\
}}


%


\maketitle

\begin{abstract}
In this paper, we propose a Singular-Value-Decomposition-based variable-resolution Analog to Digital Converter (ADC) bit allocation design for a single-user Millimeter wave massive Multiple-Input Multiple-Output  receiver. We derive the optimality condition for bit allocation under a power constraint. This condition ensures optimal receiver performance in the Mean Squared Error (MSE) sense. We derive the MSE expression and show that it approaches the Cramer-Rao Lower Bound (CRLB). The CRLB is seen to be a function of the analog combiner, the digital combiner, and the bit allocation matrix. We attempt to minimize the CRLB with respect to the bit allocation matrix by making suitable assumptions regarding the structure of the combiners. In doing so, the bit allocation design reduces to a set of simple inequalities consisting of ADC bits, channel singular values and covariance of the quantization noise along each RF path. This results in a simple and computationally efficient bit allocation algorithm. Using simulations, we show that the MSE performance of our proposed bit allocation is very close to that of the Full Search (FS) bit allocation. We also show that the computational complexity of our proposed method has an order of magnitude improvement compared to FS and Genetic Algorithm based bit allocation of $\cite{Zakir1}$.
\end{abstract}


%

%
\IEEEpeerreviewmaketitle

\section{Introduction}
Hybrid precoding and combining is a common architecture used with Millimeter wave (mmWave) Massive Multiple-Input Multiple-Output (MIMO) transceivers $\cite{hybrid,mmPreCom,SigProc}$. Also, in mmWave massive MIMO receivers, low-resolution combining is a popular technique  to deal with large power demands by the Analog to Digital Converters (ADC) $\cite{mmPreCom,SigProc,5GBackHaul}$. In this paper, we propose to combine the SVD-based combining with low-resolution ADCs at the receiver.

%
 In an earlier work, it was shown that for certain channel conditions, the optimal bit allocation across the RF paths is non-uniform and  channel-dependent $\cite{Zakir1}$. Hence, we use a generic structure to allow non-uniform bit allocation. Our focus in this work is primarily on the design of ADC bit allocation at the receiver and hybrid combiners via channel SVD. We assume that the hybrid precoder design is independent of combiners and ensure that the optimal hybrid precoders match the right singular matrix of the MIMO channel $\cite{SigProc,PreDsgn}$.

We derive the expression for the MSE of the received, quantized, and combined vector. Our contributions in this paper are: (1) we derive the expression for CRLB and show that the MSE achieves CRLB, and (2) we attempt to minimize the CRLB with respect to the bit allocation matrix. The bit-allocation matrix enable variable bit-allocation on the receivers RF-paths. In doing so, we arrive at the conditions for hybrid combiners and a simple algorithm for bit allocation with order of magnitude reduction in computational complexity. We present the simulation results by computing MSE at various SNRs using the proposed bit allocation and Full Search (FS) technique. We simulate the mmWave channel using the simulator of Shu Sun et. al $\cite{nyusim}$ with 8 and 12 RF Paths. We use 64 Rx and 32 Tx antenna elements arranged in a Uniform Linear Array (ULA).

\section{Signal Model}
The signal model for a typical single-user mmWave Massive MIMO transceiver encompassing a hybrid precoding and combining is shown in Figure \ref{fig:Fig1new.pdf}. Here, we denote ${\bold{F}_D}$ and ${\bold{F}_A}$ to be the digital and analog precoders, respectively. Similarly, we represent  ${\bold{W}_D^H}$ and ${\bold{W}_A^H}$ to be the digital and analog combiners, respectively. The vector $\bold{x}$ is an $Ns\times1$ transmitted signal vector whose average power is unity. Let $N_{rt}$ and $N_{rs}$ denote the number of RF Chains at the transmitter and  receiver, respectively. Also, $N_t$ and $N_r$ represent the number of transmit and receive antennas, respectively. The channel matrix $\bold{H} = \big[ h_{ij} \big]$ is an \begin{math}N_r\times N_t\end{math}  matrix representing the line of sight  mmWave MIMO channel with properties defined in $\cite{rapaport}$. 

\begin{figure}[h]
\begin{center}
\input{Fig1.latex}
\caption{Signal Model}
\label{fig:Fig1new.pdf}
\end{center}
\end{figure}
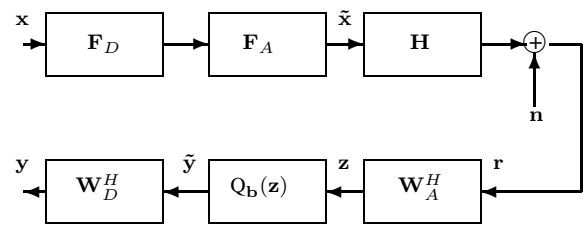

The transmitted signal $\bold{\tilde{x}}$ and the received signal $\bold{r}$ are represented as 
\begin{equation}\label{eq1a}
\begin{split}
{\bold{\tilde{x}}} &= {\bold{F}_A}{\bold{F}_D}{\bold{x}}, \\
{\bold{r}} &= {\bold{H}}{\bold{\tilde{x}}}+{\bold{n}}.
\end{split}
\end{equation}

Here, ${\bold{n}}$ is an $N_r\times1$ noise vector of independent and identically distributed (i.i.d) complex Gaussian random variables such that ${\bold{n}} \sim \mathcal{CN}(\bold{0},{\sigma_n^2}{\bold{I}_{N_r}})$. The received symbol vector $\bold{r}$ is analog-combined with ${\bold{W}_A^H}$ to get ${\bold{z}} = {\bold{W}_A^H}{\bold{r}}$  and later  digitized using a variable-bit quantizer to produce ${\bold{\tilde{y}}} = \text{Q}_{\bold{b}} \big( {\bold{z}} \big) = \bold{W}_{\alpha}\big( {\bold{b}} \big){\bold{z}}+{\bold{n_q}}$. This  signal is later combined using the digital combiner ${\bold{W}_D^H}$ to produce the output signal ${\bold{y}} = {\bold{W}_D^H}{\bold{\tilde{y}}}$. The quantizer is modeled as an Additive Quantization Noise Model (AQNM) $\cite{Rangan,Uplink}$. However, when we extend this model to allocate unequal ADC bits (1-4) across $N_{s}$ RF paths, the AQNM model $\text{Q}_{\bold{b}} \big( {\bold{z}} \big)$ can be written as $\text{Q}_{\bold{b}} \big( {\bold{z}} \big) = \bold{W}_{\alpha}\big( {\bold{b}} \big){\bold{z}}+{\bold{n_q}}$ $\cite{Zakir1,VarBitAlloc}$.  Vector $\bold{n}_q$ is the additive quantization noise which is uncorrelated with $\bold{z}$ and has a Gaussian distribution: ${\bold{n}_q} \sim \mathcal{CN}(\bold{0},{\bold{D}_q^2})$ $\cite{Rangan,Uplink,VarBitAlloc}$.


Hence, the relationship between the transmitted signal vector $\bold{x}$ and the received symbol vector $\bold{y}$ at the receiver is given by
\begin{equation}\label{eq5a}
\begin{split}
{\bold{y}} &= {\bold{W}_D^H}{\bold{W}_{\alpha}}{\big( {\bold{b}} \big)}{\bold{W}_A^H}{\bold{H}}{\bold{F}_A}{\bold{F}_D}{\bold{x}} + {\bold{W}_D^H}{\bold{W}_{\alpha}\big( {\bold{b}} \big)}{\bold{W}_A^H}{\bold{n}} \\
&+{\bold{W}_D^H}{\bold{n_q}}.
\end{split}
\end{equation}

The  dimensions of matrices indicated in Figure $\ref{fig:Fig1new.pdf}$ are as follows:
${\bold{F}_D} \in \mathbb{C}^{N_{rt} \times N_s}$, ${\bold{F}_A} \in \mathbb{C}^{N_t \times N_{rt}}$, ${\bold{H}} \in \mathbb{C}^{N_r \times N_t}$, ${\bold{W}_A^H} \in \mathbb{C}^{N_{rs} \times N_r}$, ${\bold{W}_D^H} \in \mathbb{C}^{N_s \times N_{rs}}$ and $\bold{W}_{\alpha}\big( {\bold{b}} \big) \in \mathbb{R}^{N_{rs} \times N_{rs}}$.\\

We have $\bold{W}_{\alpha}\big( {\bold{b}} \big)$ as the diagonal bit allocation matrix. We intend to design the precoders  ${\bold{F}_D}$ and ${\bold{F}_A}$, and Combiners ${\bold{W}_D^H}$ and ${\bold{W}_A^H}$, along with the ADC bit allocation ${\bold{W}_{\alpha}\big( {\bold{b}} \big)}$ for a given channel realization $\bold{H}$. We assume that we have perfect CSI at the transmitter. We further assume that the number of RF paths $N_{rs}$ on the receiver is the same as the number of parallel data streams $N_s$, i.e., $N_{rs} = N_s$. The analysis is easy to extend and similar for the case $N_{rs} \ne N_s$.
\subsection{Precoder Design}
The hybrid precoding and combing techniques for systems employing phase shifters in mmWave transceiver architectures impose some constraints on them. The constraint is to  have  unit norm entries in the analog precoders and combiners, that is $\bold{F}_A$ and ${\bold{W}_A^H}$, respectively. Given the number of constraints we have on the designing these matrices; solving them is quite complicated. Hence, we propose to design the precoder and combiner separately. For the precoder design, we employ the technique used in $\cite{SigProc,PreDsgn}$. Let the SVD of the channel matrix $\bold{H}$ be
\begin{equation}\label{eq7a}
\begin{gathered}
\bold{H} = \bold{U}\bold{\Sigma}\bold{F}_{\text{opt}}^H,\\
{\text{where }} {\bold{U}} \in \mathbb{C}^{N_r \times N_s}, {\bold{\Sigma}} \in \mathbb{R}^{N_s \times N_s}, {\bold{F}_{\text{opt}}} \in \mathbb{C}^{N_t \times N_s}.
\end{gathered}
\end{equation}

The hybrid precoders are derived upon solving the optimization problem $\cite{PreDsgn}$ stated below. 
\begin{equation}\label{eq8a}
\begin{aligned}
({\bold{F}_A^{opt}},{\bold{F}_D^{opt}}) = & \underbrace{\text{argmin}}_{{\bold{F}_D},{\bold{F}_A}}{\lVert {{\bold{F}_{\text{opt}}} - {{\bold{F}_A}{\bold{F}_D}}} \rVert }_F, \\
\text{such that } & {\bold{F_A}}\in{\mathcal{F}_{RF}}, {\lVert {{\bold{F}_D}{\bold{F}_A}} \rVert }_F^2 = N_s.
\end{aligned}
\end{equation}

The set $\mathcal{F}_{RF}$ is the set of all possible analog precoders that correspond to hybrid precoder architecture based on phase shifters. This includes all possible $N_t \times N_{rt}$ matrices with constant magnitude entries.

\subsection{Combiner and Bit-allocation Design}
We now design the combiners such that the MSE of the received, quantized and combined vector is minimized. Having designed the precoders in the previous section such that ${\bold{F}_{\text{opt}}} \approx {\bold{F}_A}{\bold{F_D}}$ with the constraints in ($\ref{eq8a}$), we can rewrite ($\ref{eq5a}$) as 
\begin{equation}\label{eq9a}
\begin{aligned}
\begin{split}
{\bold{y}} &= {\bold{W}_D^H}{\bold{W}_{\alpha}}{\big( {\bold{b}} \big)}{\bold{W}_A^H}{\bold{U}}{\bold{\Sigma}}{\bold{x}} + {\bold{W}_D^H}{\bold{W}_{\alpha}\big( {\bold{b}} \big)}{\bold{W}_A^H}{\bold{n}} \\
&+{\bold{W}_D^H}{\bold{n_q}}.
\end{split}
\end{aligned}
\end{equation}
Using ($\ref{eq9a}$) we derive the expression for $\mbox{MSE}$ $\delta$ as
\begin{equation}\label{10aa}
\begin{gathered}
\delta \triangleq \tr{(E\big[ (\bold{y}-\bold{x})^2\big])}\\
\mbox{MSE}(\bold{x}) = E\big[ (\bold{y}-\bold{x})^2\big]\\
\mbox{MSE}(\bold{x}) = p({\bold{K}}-{\bold{I}_{N_s}})^2 + {\sigma_{n}^2}{\bold{G}\bold{G}^H} + {\bold{W}_D^H}{\bold{D}_q^2}{\bold{W}_D},
\end{gathered}
\end{equation}
where
${\bold{K}} = {\bold{W}_D^H}{\bold{W}_{\alpha}}{\bold{W}_A^H}{\bold{U}}{\bold{\Sigma}}, 
E[{\bold{x}}{\bold{x}}^H] = p{\bold{I}_{N_s}}, 
{\bold{G}} = {\bold{W}_D^H}{\bold{W}_{\alpha}}{\bold{W}_A^H}, 
E[{\bold{n}}{\bold{n}}^H] = {\sigma_n^2}{\bold{I}_{N_r}}, 
E[{\bold{n_q}}{\bold{n_q}}^H] = {\bold{D}_q^2}, 
{\bold{D}_q^2} = {\bold{W}_{\alpha}}{\bold{W}_{1-\alpha}}{\text{diag}}[ {\bold{W}_A^H}{\bold{H}}({\bold{W}_A^H}{\bold{H}})^H+{\bold{I}_{N_{rs}}} 
],
E[{\bold{n}}{\bold{n_q}}^H] = 0.$\\

\indent
We intend to design the combiners $\bold{W}_A^H$ and $\bold{W}_D^H$ and the bit allocation matrix ${\bold{W}_{\alpha}}{\big( {\bold{b}} \big)}$ such that the MSE in ($\ref{10aa}$) is minimized. Thus we set ${\bold{K}} = {\bold{I}_{N_s}}$, which gives ${\bold{G}\bold{G}^H} = {\bold{\Sigma}^{-2}}$. Hence, the $\mbox{MSE}(\bold{x})$ in ($\ref{10aa}$) is reduced to:
\begin{equation}\label{10ac}
\mbox{MSE}(\bold{x}) = {\sigma_n^2}{\bold{\Sigma}^{-2}} + {\bold{W}_D^H}{\bold{D}_q^2}{\bold{W}_D}.
\end{equation}

The first term of $\mbox{MSE}(\bold{x})$ in ($\ref{10ac}$) is channel-dependent, and the only design parameter we have in our control to minimize $\mbox{MSE}(\bold{x})$ is to ensure that the second term ${\bold{W}_D^H}{\bold{D}_q^2}{\bold{W}_D}$ vanishes. Thus, the combiners and bit allocation matrix needs to be designed with the condition given as:
\begin{equation}\label{eq11a}
\begin{aligned}
\begin{split}
{\bold{K}} = {\bold{W}_D^H}{\bold{W}_{\alpha}}{\big( {\bold{b}} \big)}{\bold{W}_A^H}{\bold{U}}{\bold{\Sigma}} &= \bold{I}_{N_s},\\
\text{such that }{\bold{W}_D^H}{\bold{D}_q^2}{\bold{W}_D} &= \bold{0}.
\end{split}
\end{aligned}
\end{equation}

This is a hard problem to solve for $\bold{W}_A^H$, $\bold{W}_D^H$, and ${\bold{W}_{\alpha}}{\big( {\bold{b}} \big)}$. Thus, we take a slightly different approach. We will prove  (theorem \ref{Thm1}) that ($\ref{10ac}$) is indeed the Minimum MSE (MMSE) that can be achieved for a given $\bold{W}_A^H$, $\bold{W}_D^H$, ${\bold{W}_{\alpha}}{\big( {\bold{b}} \big)}$, and channel $\bold{H}$. This is accomplished by deriving the expression for the CRLB for estimating $\bold{x}$, given the observation $\bold{y}$. Secondly, we view the CRLB as a function of design parameters $\bold{W}_A^H$, $\bold{W}_D^H$, and ${\bold{W}_{\alpha}}{\big( {\bold{b}} \big)}$, and minimize the CRLB to design them \cite{MinCRLB}.
\newtheorem{theorem}{Theorem}
\begin{theorem}\label{Thm1}
If $\mbox{MSE}(\bold{x})$ is the Mean Square Error matrix as defined in ($\ref{10ac}$), and if ${\bold{\hat{x}}}$ is the estimate of $\bold{x}$ given the observation $\bold{y}$ in ($\ref{eq9a}$) for a given fixed $\bold{H}$, $\bold{W}_A^H$, $\bold{W}_D^H$, ${\bold{W}_{\alpha}}{\big( {\bold{b}} \big)}$, then there exists an estimator that is efficient. That is, $\mbox{MSE}(\bold{x})$ achieves the CRLB ${\bold{I}^{-1}({\bold{\hat{x}}})}$ under similar conditions, in other words $\mbox{MSE}(\bold{x})$ is indeed MMSE.
\end{theorem}
\begin{proof}
We look at the problem in ($\ref{eq9a}$) as an Estimation problem, such that we need to estimate $\bold{x}$ given $\bold{y}$ is observed; given that $\bold{W}_A^H$, $\bold{W}_D^H$, ${\bold{W}_{\alpha}}{\big( {\bold{b}} \big)}$ are fixed. For simplicity of notation, we will write ${\bold{W}_{\alpha}}{\big( {\bold{b}} \big)}$ as $\bold{W}_{\alpha}$. Now, we will rewrite the ($\ref{eq9a}$) as 
\begin{equation}\label{eq13a}
{\bold{y}} = {\bold{K}}{\bold{x}} + {\bold{n_1}}
\end{equation}
where ${\bold{K}} = {\bold{W}_D^H}{\bold{W}_{\alpha}}{\bold{W}_A^H}{\bold{U}}{\bold{\Sigma}}$, and 
$\bold{n_1} = {\bold{W}_D^H}{\bold{W}_{\alpha}}{\bold{W}_A^H}{\bold{n}} + {\bold{W}_D^H}{\bold{n_q}}.$
We know that $\bold{n}$ and $\bold{n_q}$ are Gaussian random vectors, with the following statistics:
\vspace{-7.3mm}
\begin{center}
\begin{equation}\label{eq16a}
\begin{split}
&\bold{n} \sim \mathcal{N}(\bold{0},{\sigma_n^2}{\bold{I}_{N_r}}),\\
&\bold{n_q} \sim \mathcal{N}(\bold{0},{\bold{D}_q^2}).
\end{split}
\end{equation}
\end{center}
The statistical distribution of $\bold{n_1}$ is given by: 
\begin{equation}\label{eq17a}
E[\bold{n_1}] = {\bold{W}_D^H}{\bold{W}_{\alpha}}{\bold{W}_A^H}E[{\bold{n}}] + {\bold{W}_D^H}E[{\bold{n_q}}] = {\bold{0}},
\end{equation}
\begin{equation}\label{eq18a}
\begin{split}
{\sigma_{n_1}^2} &= E[ (\bold{n_1} - E[\bold{n_1}])^2 ] = E[{\bold{n_1}}{\bold{n_1}}^H]\\
&= {\sigma_n^2}{\bold{G}}{\bold{G}^H} + {\bold{W}_D^H}{\bold{D}_q^2}{\bold{W}_D}.
\end{split}
\end{equation}
%
%
Thus, the statistics of $\bold{n_1}$ follows: 
\begin{equation}\label{eq22a}
\bold{n_1} \sim \mathcal{N}(\bold{0},({{\sigma_n^2}{\bold{G}}{\bold{G}^H} + {\bold{W}_D^H}{\bold{D}_q^2}{\bold{W}_D}})).
\end{equation}

It is noted that ${\bold{G}}$ is an $N_s \times N_r$ matrix with $N_r \gg N_s$. It is safe to assume that $\bold{G}$ has a full row Rank and its pseudo inverse exists $\cite{LAStrang}$.

Equation ($\ref{eq13a}$) can be seen as a linear model, in which we intend to estimate $\bold{x}$, given the observation $\bold{y}$. We can express the conditional probability distribution of ${\bold{y}}$ given ${\bold{x}}$ as $\cite{Kay}$
\begin{equation}\label{eq23a}
p({\bold{y} \vert {\bold{x}}}) \sim \frac{1}{({2\pi}{{\sigma_{n_1}^2}})^{\frac{N_s}{2}}} \text{exp} \bigg\{ -\frac{1}{2{\sigma_{n_1}^2}} ({\bold{y}}-{\bold{K}}{\bold{x}})^H({\bold{y}}-{\bold{K}}{\bold{x}}) \bigg\}.
\end{equation}

From ($\ref{eq13a}$) and ($\ref{eq23a}$), it is straightforward to see that the ``regularity conditions" are satisfied, and hence for such a linear estimator, we can write $\cite{Kay}$ the expression for the CRLB as 
\begin{equation}\label{eq24a}
{\bold{I}^{-1}({\bold{\hat{x}}})} = ({\bold{K}^H}{\bold{C}^{-1}}{\bold{K}})^{-1},
\end{equation}
where C is the noise covariance matrix of ${\bold{n_1}}$ as given in ($\ref{eq22a}$).
Substituting $\bold{K}$ and $\bold{C}$ in ($\ref{eq24a}$), we arrive at
\begin{equation}\label{eq26a}
\begin{split}
{\bold{I}^{-1}({\bold{\hat{x}}})} &= ({\bold{K}^H}{\bold{C}^{-1}}{\bold{K}})^{-1}\\
&={\sigma_n^2}{\bold{\Sigma}^{-2}}+{\bold{K}^{-1}}{\bold{W}_D^H}{\bold{D}_q^2}{\bold{W}_D}({\bold{K}^H})^{-1}.
\end{split}
\end{equation}


Now, on setting ${\bold{K}} = {\bold{I}_{N_s}}$, the expression for the CRLB is simplified to 
\begin{equation}\label{eq27a}
{\bold{I}^{-1}({\bold{\hat{x}}})} = {\sigma_n^2}{\bold{\Sigma}^{-2}}+{\bold{W}_D^H}{\bold{D}_q^2}{\bold{W}_D}.\end{equation}
\end{proof}
\renewcommand{\qed}{\hfill\blacksquare}
Thus, for a given ${\bold{W}_A^H}$, ${\bold{W}_D^H}$ and $\bold{W}_\alpha$, we see that the expression for CRLB in ($\ref{eq27a}$) is the same as the $\mbox{MSE}(\bold{x})$ in ($\ref{10ac}$). Hence, $\mbox{MSE}(\bold{x})$ in ($\ref{10ac}$) is indeed MMSE. $\qed$ \\


\subsubsection{Minimizing the CRLB}
Given the fact that the MMSE derived using  ($\ref{10ac}$) achieves CRLB for fixed design parameters, we now intend to design the combiners $\bold{W}_A$, $\bold{W}_D$ and bit allocation ${\bold{W}_{\alpha}}$ by minimizing the CRLB \cite{MinCRLB}. Referring to ($\ref{eq26a}$), we wish to have ${\bold{I}^{-1}({\bold{\hat{x}}})}$  vanish or gets close to zero:
\begin{equation}\label{eq28a}
{\bold{I}^{-1}({\bold{\hat{x}}})} = {\sigma_n^2}{\bold{\Sigma}^{-2}}+{\bold{K}^{-1}}{\bold{W}_D^H}{\bold{D}_q^2}{\bold{W}_D}({\bold{K}^H})^{-1} \approx \bold{0}.
\end{equation}

Substituting  $\bold{K}$  into ($\ref{eq28a}$) and simplifying the equation, we arrive at 
\begin{equation}\label{eq29a}
{\sigma_n^2}{\bold{\Sigma}^{-2}}+{\bold{\Sigma}}^{-1}{\bold{U}^H}({\bold{W}_A^H})^{-1}{\bold{W}_{\alpha}^{-1}}{\bold{D}_q^2}{\bold{W}_{\alpha}^{-1}}{\bold{W}_A^{-1}}{\bold{U}}{\bold{\Sigma}}^{-1} \approx \bold{0}.       
\end{equation}

Phase shifters or splitters impose constraints on the design of the analog combiner $\bold{W}_A^H$ \cite{SigProc}. We will denote the constrained analog combiner as $\bold{\tilde{W}}_A^H$. The imperfections in the analog combiner are compensated by the digital combiner:
\begin{equation}\label{eq30a}
{\bold{W}_A^H} = {\bold{W}_D}{\bold{\tilde{W}}_A^H}.
\end{equation}

We also  would like to design the actual analog combiner ${\bold{\tilde{W}}_A^H}$ and the digital combiner ${\bold{W}_D}$, such that:
\begin{equation}\label{eq31a}
\begin{split}
{\bold{W}_A^H} = &{\bold{U}^H} = {\bold{W}_D}{\bold{\tilde{W}}_A^H};\\
\text{or } &{\bold{U}} = {\bold{\tilde{W}}_A}{\bold{W}_D^H}.\\
\end{split}
\end{equation}

The analog and digital combiner can now be designed more easily. We start with ${\bold{{W}}_A^H} = {\bold{U}^H}$, then we adjust the ${\bold{\tilde{W}}_A^H}$ and ${\bold{W}_D}$ to meet the necessary constraints.\\ 

To design the bit allocation, we  substitute  ($\ref{eq31a}$) in ($\ref{eq29a}$) to arrive at
\begin{equation}\label{eq32a}
{\bold{I}^{-1}({\bold{\hat{x}}})} = {\bold{\Sigma}^{-2}}\bigg[ {\sigma_n^2}{\bold{I}_{N_s}} +  {\bold{W}_{\alpha}^{-2}}{\bold{D}_q^2}\bigg] \approx \bold{0}.
\end{equation}

In order to satisfy ($\ref{eq32a}$), the condition for the bit allocation becomes:
\begin{equation}\label{eq33a}
{\bold{\Sigma}}^{2} \gg {\sigma_n^2}{\bold{I}_{N_s}} +  {\bold{W}_{\alpha}^{-2}}{\bold{D}_q^2}.
\end{equation}
Since ${\bold{\Sigma}}^2$, ${\bold{W}_{\alpha}^2}$ and ${\bold{D}_q^2}$ are diagonal matrices, we can rewrite $(\ref{eq33a})$ as a set of $N_s$ inequalities:  
\begin{equation}\label{eq34a}
\begin{gathered}
{{\sigma_i}^2} \gg \sigma_n^2 + \frac{f(b_i)}{1-f(b_i)}({1+{\bold{w}_{A_i}^H}{\bold{h_i}^H}{\bold{h_i}}{\bold{w}_{A_i}}}),\\
\text{ for }1 \le i \le N_s,
\end{gathered}
\end{equation}
where ${\sigma_i}$ is the diagonal element of ${\bold{\Sigma}}$, ${\sigma_n^2}$ is the noise power, $f(b_i)$ is the ratio of the Mean Square Quantization Error (MQSE) and the power of the symbol for a non-uniform MMSE quantizer with $b_i$ bits along the RF path $i$, $i=1,2,..N_s$ $\cite{VarBitAlloc}$. The values for $f(b_i)$ are indicated in the Table $\ref{betaVal}$. $\bold{w}_{A_i}$ and $\bold{h}_i$ are the $i^{th}$ column of the matrix $\bold{W_A}$ and $\bold{H}^H$, respectively.
\begin{table}[http]
\begin{tabu} to 0.5\textwidth { c c c c c c }
\hline
$b_i$  & 1 & 2 & 3 & 4 & 5 \\
\hline
$f(b_i)$ & 0.3634 & 0.1175 & 0.03454 & 0.009497 & 0.002499 \\
\hline
\end{tabu}
\vspace{1mm}
\caption{Values of $f(b_i)$ for different ADC Quantization Bits $b_i$} \label{betaVal}
\end{table}
\vspace{-5mm}
\subsection{Algorithm Description for bit allocation}
We need to satisfy all the $N_s$ inequalities in ($\ref{eq34a}$) to attain the optimal bit allocation. However, it would not be possible to attain optimal bit allocation, given the number of bits and the power budget. In such scenarios, we would make a best effort approach to satisfy the set of equations in ($\ref{eq34a}$) and the solution would be the best solution given the bits and power constraints.

For a given $N_s$, we first form a super set $B_{\text{set}}$ of all possible bit allocations $\cite{Zakir1}$ that satisfy the given ADC power budget $P_{\text{ADC}}$ as
\begin{equation}\label{eq35a}
\begin{split}
B_{\text{set}} \triangleq \big\{ &\bold{b}_j = {\big[ b_{j1}, b_{j2}, \dots, b_{jN}  \big]}^T \text{ for } 0 \leq j < 4^{N_s} \mid \\
& 1 \le b_{ji} \le 4 \text{ and } \sum_{i=1}^{N} cf_s2^{b_{ji}} \leq P_{\text{ADC}} \big\}.
\end{split}
\end{equation}

We incorporate a gain term ${K_{f}}(b_i)$ for a given bit $b_i$ on RF path $i$ into the set of equations in ($\ref{eq34a}$) such that
\begin{equation}\label{eq36a}
{K_{f}}(b_i) \triangleq \frac{{\sigma_i^2}}{\sigma_n^2+\frac{f(b_i)}{1-f(b_i)}({1+{\bold{w}_{A_i}^H}{\bold{h_i}^H}{\bold{h_i}}{\bold{w}_{A_i}}})}.
\end{equation}
For a given bit allocation $\bold{b}_j$ in $B_{\text{set}}$, we denote
\begin{equation}\label{eq37a}
{K_{f}}(\bold{b}_j) \triangleq \sum_{i=1}^{N_s} \Bigg[\frac{{\sigma_i^2}}{\sigma_n^2+\frac{f(b_i)}{1-f(b_i)}({1+{\bold{w}_{A_i}^H}{\bold{h_i}^H}{\bold{h_i}}{\bold{w}_{A_i}}})}\Bigg]. 
\end{equation}
We select ${\bold{b}} \in B_{\text{set}}$ to maximize ${K_{f}}(\bold{b}_j)$ and declare that as the desirable bit allocation solution. The Algorithm is described in Algorthm~$\ref{AlgoCRLB}$.
\begin{algorithm}
  \caption{CRLB-based Bit Allocation}\label{AlgoCRLB}
  \begin{algorithmic}
    \Procedure{CRLB-based Bit Allocation}{$\bold{H},\bold{W}_A,\bold{\Sigma},B_{\text{set}},N_s, f(\cdot),\sigma_n^2$}
      \State $\bold{H}\gets \text{MIMO channel}$
      \State $\bold{W}_A\gets \text{Combiners designed as per ($\ref{eq31a}$)}$
      \State $\bold{\Sigma} \gets \text{Matrix containing singular values $\sigma_i$}$
      \State $B_{\text{set}}\gets \text{Bit allocations adhering to ADC power budget}$
      \State $N_s\gets \text{Number of spatial-multiplexed paths}$
      \State $f(\cdot)\gets \text{Quantization error lookup table}$
      \State $\sigma_n^2\gets \text{AWG Noise power}$
      \For{\texttt{j=0;j++ ;until j<size of $B_{\text{set}}$}}
           \State $K_f(b_j) = 0$
           \For{\texttt{i=0;i++ ;until i<$N_s$}}
                 \State \footnotesize{$K_f(b_j)\gets K_f(b_j) + \frac{{\sigma_i^2}}{\sigma_n^2+\frac{f(b_i)}{1-f(b_i)}({1+{\bold{w}_{A_i}^H}{\bold{h_i}^H}{\bold{h_i}}{\bold{w}_{A_i}}})}$}
                 \normalsize
           \EndFor
     \EndFor
     \State $index\gets \text{max}(K_f)$
     \State $\bold{b}\gets B_{\text{set}}\text{ at } index$
     \State \textbf{return} $\bold{b}$ \Comment{Optimal Bit Allocation Vector}
    \EndProcedure
  \end{algorithmic}
\end{algorithm}

\section{Test Setup and Simulation Results}
We simulate the mmWave channel using the NYUSIM channel simulator with the configurations specified in Table $\ref{nyusimtab}$ $\cite{nyusim}$. We strengthen the singular value on the dominant channel to simulate a strong scatterer $\cite{SVDcorrChannels}$. With this simulated channel, we use $N_s=8$ or $N_s=12$ strong channels (RF paths) to spatially multiplex 64 QAM data symbols.\\
\indent
With the above channel, we run the FS technique to arrive at the optimal bit allocation and compute the MSE $\delta$. We do this at various SNRs. This simulation is shown in pink for $N_s = 8$ and $N_s = 12$ in Figure $\ref{fig:crlb_Nr8_H64by32.eps}$ and Figure $\ref{fig:crlb_Nr12_H64by32.eps}$, respectively. Similarly, we run the simulation to arrive at the bit allocations using our proposed approach. The simulation of MSE $\delta$ at various SNRs are plotted in black. In both situations, we set the  ADC power budget to be the power consumed on having 2-bit ADCs on all RF paths. We constraint the ADC bit resolution between 1 and 4 bits. Figures $\ref{fig:crlb_Nr8_H64by32.eps}$ and $\ref{fig:crlb_Nr12_H64by32.eps}$ also has simulations of MSE performance at various SNRs on having all 2-bit ADCs (red plots), all 1-bit ADCs (blue plots), and no quantization (green plots).
\begin{figure}[t!]
\centering
\includegraphics[width=0.5\textwidth]{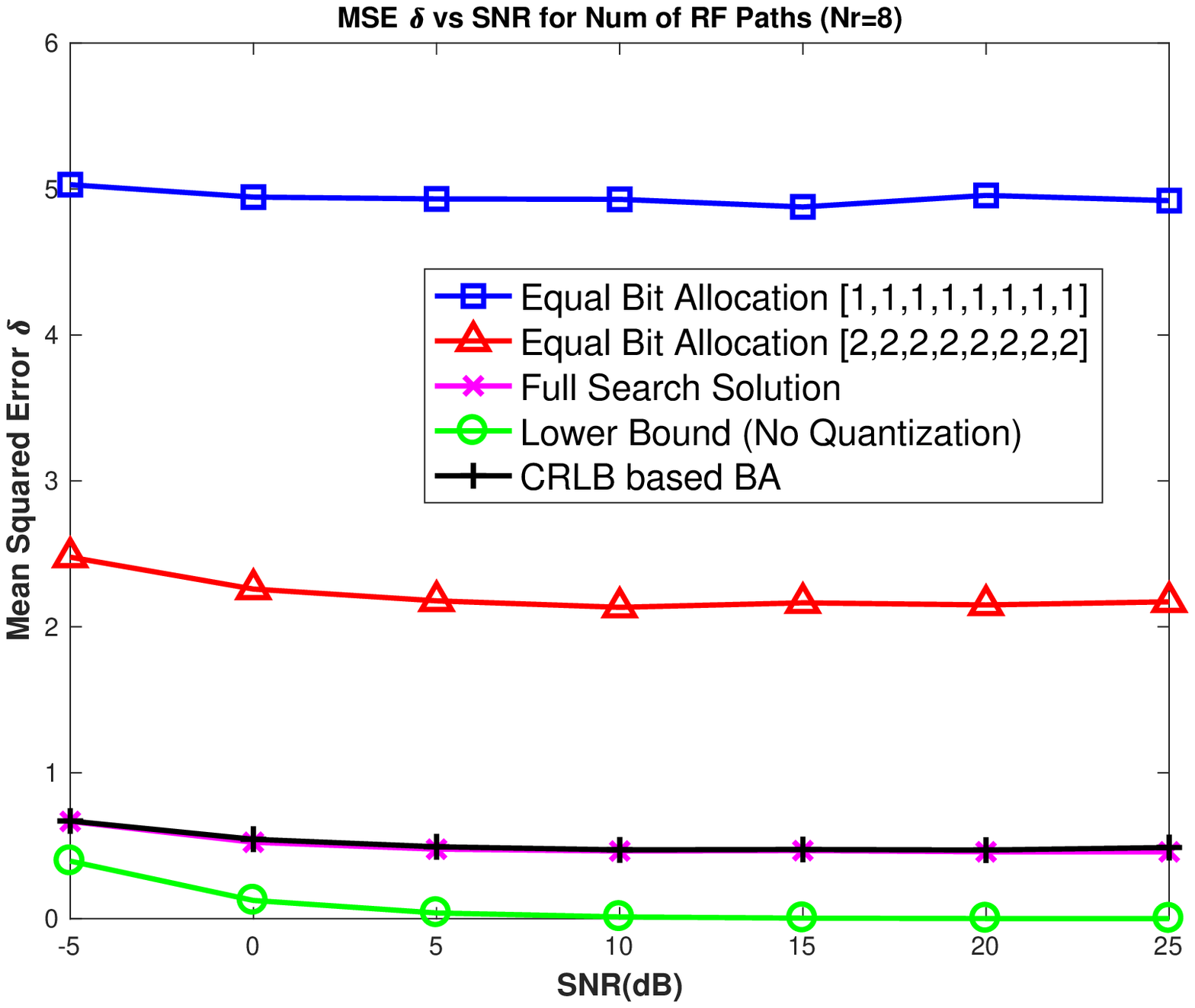}
\caption{MSE $\delta$ vs. SNR for $N_s=8$ for all 1-bit, 2-bit, FS and CRLB-based simulation runs}
\label{fig:crlb_Nr8_H64by32.eps}
\qquad
\includegraphics[width=0.5\textwidth]{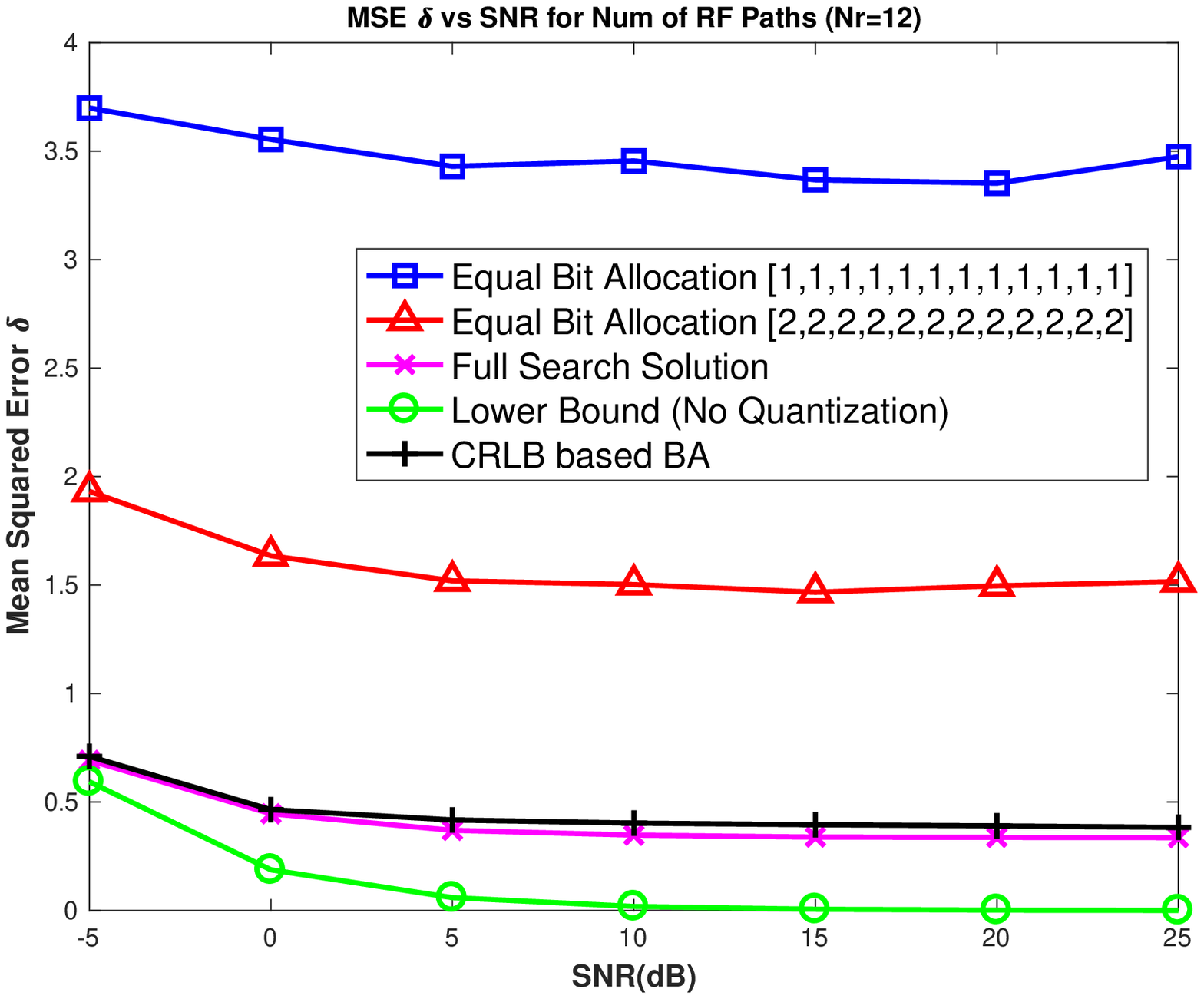}
\caption{MSE $\delta$ vs. SNR for $N_s=12$ for all 1-bit, 2-bit, FS and CRLB-based simulation runs}
\label{fig:crlb_Nr12_H64by32.eps}
\end{figure}
\begin{table}
\begin{center}
\begin{tabu} to 0.5\textwidth {| l| l| }
 \hline
 \textbf{Parameters}  & \textbf{Value/Type} \\
 \hline
Frequency & 28Ghz \\
\hline
Environment & Line of sight \\
\hline
T-R seperation & 100m\\
\hline
TX/RX array type & ULA\\
\hline 
Num of TX/RX elements $N_t$/$N_r$ & 32/64\\
\hline
TX/RX  antenna spacing & $\lambda/2$\\
\hline
\end{tabu}
\vspace{1mm}
\caption{$\text{Channel parameters for NYUSIM model $\cite{nyusim}$}$} \label{nyusimtab}
\vspace{-3mm}
\end{center}
\end{table}
Table \ref{tab:CRLBTab} and Table \ref{tab:CRLBTab1} summarizes the computational complexities of FS, GA and our proposed algorithm. Table \ref{tab:CRLBTab} captures the number of evaluations of MSE $\delta$ in ($\ref{10aa}$) required for FS Algorithm and GA, and the number of evaluations of gains ${K_{f}}(\bold{b}_j)$ in ($\ref{eq37a}$) for CRLB-based algorithm to arrive at the solution.\\
\indent
The evaluation of MSE $\delta$ requires computation of $\bold{y}$ in ($\ref{eq9a}$) which requires multiple matrix multiplications with signal vector of order $N_s$ in the least. Hence, computational complexity of MSE $\delta$ is $O(N_s^2)$. However, the computation of ${K_{f}}(\bold{b}_j)$ in ($\ref{eq37a}$) requires 2 dot product evaluations of size $N_t$ and $N_r$, $N_s$ times. Hence its computational complexity is $O(N_s)$. Thus, the CRLB-based bit allocation brings down the computational complexity compared to the FS algorithm and GA by an order of magnitude.
\begin{table}
\begin{center}
\begin{tabu} to 0.45\textwidth { | X[c] | X[c] | X[c] | X[c] |}
\hline
\small Number of RF paths  & \small FS algorithm & \small GA & \small CRLB-based algorithm \\
\hline
  & \small Num. of MSE evaluation $\gamma$. $O(N_s^2)$  & \small Num. of MSE evaluation $\gamma$. $O(N_s^2)$ & \small Num. of  ${K_{f}}(\bold{b}_j)$ evaluations $\mu$. $O(N_s)$\\
\hline
\small8 & \small1878 & \small324 & \small1878 \\
\hline
\small12 & \small133253 & \small2025 & \small133253 \\
\hline
\end{tabu}
\vspace{2mm}
\caption{Number of MSE $\delta$ evaluations of FS, GA and number of $K_f$ evaluations of CRLB-based algorithm} \label{tab:CRLBTab}
\begin{tabu} to 0.47\textwidth { | p{0.75cm} | p{0.95cm} | X[c] | p{0.6cm} | p{0.95cm} | X[c] | p{0.8cm} |}
\hline
\multirow{3}{1.0cm}{\footnotesize Number of RF paths} & \multicolumn{3}{c|}{\footnotesize Number of complex} & %
    \multicolumn{3}{c|}{\footnotesize Number of complex}\\
& \multicolumn{3}{c|}{\footnotesize multiplications$^{\S}$} & \multicolumn{3}{c|}{\footnotesize additions$^{\S}$}  \\
\cline{2-7}
& \centering \footnotesize FS & \footnotesize GA & \centering \footnotesize CRLB-based & \centering \footnotesize FS & \footnotesize GA & \centering \footnotesize CRLB-based\\
\hline
\centering \multirow{2}{*}{\small 8}  & \scriptsize1622592 & \scriptsize 279936  & \scriptsize \textcolor{red}{864} & \scriptsize1592544 & \scriptsize 274752 & {\scriptsize \textcolor{red}{760}} \\
& & & & & & {\scriptsize \textcolor{red}{13146$^{\dagger}$}}\\
\hline
\centering \multirow{2}{*}{\small 12} & \scriptsize179092032 &\scriptsize2721600 & \scriptsize \textcolor{red}{1296} & \scriptsize175893960 & \scriptsize 2673000 & \scriptsize \textcolor{red}{1140} \\
& & & & & & {\scriptsize \textcolor{red}{1465783$^{\dagger}$}}\\
\hline
\end{tabu}

\scriptsize{$^{\S}$ See Appendix below for details. $^{\dagger}$ Real additions} \\
\caption{Computational complexity in terms of total number of multiplications and additions} \label{tab:CRLBTab1}
\end{center}
\end{table}
\section{Conclusion}
In this paper, we derive an  ADC bit allocation algorithm based on  channel SVD computation. This is derived based on the minimization of the CRLB expression seen to be a function of the hybrid combiners and bit allocation matrices. By simulation of the mmWave channel based on $\cite{nyusim}$, we observe  that the MSE performance of the receiver based on the proposed bit allocation algorithm matches very closely with that of the FS method. Comparing the computational complexity of the proposed bit allocation method with FS and GA methods, it is observed that the proposed CRLB-based technique has significant computational advantage with an MSE performance matching the FS technique.
\section*{Acknowledgment}
The authors would like to thank National Instruments for the support extended to this work.  



%
\bibliographystyle{IEEEtran}
\bibliography{Pap2BibTexFile}
\appendix
\label{FirstAppendix}
It can be shown that the total number of complex multiplications and additions required by FS and GA are $\gamma \big(N_s^2 + 4N_s + N_s(N_t + N_r) \big)$ and $\gamma \big(N_s^2 + 2N_s + N_s(N_t + N_r) \big)$ respectively. Also, $\gamma$ is the number of MSE $\delta$ evaluations.\\
\indent
Similarly, it can be shown that the total number of complex multiplications and additions required by our proposed approach are $N_s(3N_b + N_t + N_r)$ and $N_s(N_t + N_r - 1) + \mu \big(N_s-1 \big)$ respectively. Here, $N_b$ is the number of ADC bits resolution range and $\mu$ is the number of evaluations of ${K_{f}}(\bold{b}_j)$. 

\end{document}

%% file: Fig1.latex
\setlength{\unitlength}{3947sp}%
\begingroup\makeatletter\ifx\SetFigFont\undefined%
\gdef\SetFigFont#1#2#3#4#5{%
  \reset@font\fontsize{#1}{#2pt}%
  \fontfamily{#3}\fontseries{#4}\fontshape{#5}%
  \selectfont}%
\fi\endgroup%
\begin{picture}(3589,1349)(136,-861)
\put(3151,-511){\makebox(0,0)[lb]{\smash{{\SetFigFont{8}{9.6}{\familydefault}{\mddefault}{\updefault}{\color[rgb]{0,0,0}$\bold{r}$}%
}}}}
\thinlines
{\color[rgb]{0,0,0}\put(345, 75){\framebox(730,389){}}
}%
{\color[rgb]{0,0,0}\put(2341, 75){\framebox(729,389){}}
}%
{\color[rgb]{0,0,0}\put(2341,-849){\framebox(729,389){}}
}%
{\color[rgb]{0,0,0}\put(1367,-849){\framebox(730,389){}}
}%
{\color[rgb]{0,0,0}\put(345,-849){\framebox(730,389){}}
}%
\thicklines
{\color[rgb]{0,0,0}\put(1075,270){\vector( 1, 0){292}}
}%
{\color[rgb]{0,0,0}\put(2097,270){\vector( 1, 0){244}}
}%
{\color[rgb]{0,0,0}\put(3070,270){\vector( 1, 0){291}}
}%
{\color[rgb]{0,0,0}\put(3411,-120){\vector( 0, 1){341}}
}%
{\color[rgb]{0,0,0}\put(3507,270){\line( 1, 0){196}}
\put(3703,270){\line( 0,-1){925}}
\put(3703,-655){\vector(-1, 0){633}}
}%
{\color[rgb]{0,0,0}\put(2341,-655){\vector(-1, 0){244}}
}%
{\color[rgb]{0,0,0}\put(1367,-655){\vector(-1, 0){292}}
}%
{\color[rgb]{0,0,0}\put(199,270){\vector( 1, 0){146}}
}%
{\color[rgb]{0,0,0}\put(345,-655){\vector(-1, 0){146}}
}%
\thinlines
{\color[rgb]{0,0,0}\put(1367, 75){\framebox(730,389){}}
}%
\put(2551,-661){\makebox(0,0)[lb]{\smash{{\SetFigFont{8}{9.6}{\familydefault}{\mddefault}{\updefault}{\color[rgb]{0,0,0}$\bold{W}_A^H$}%
}}}}
\put(1501,-661){\makebox(0,0)[lb]{\smash{{\SetFigFont{8}{9.6}{\familydefault}{\mddefault}{\updefault}{\color[rgb]{0,0,0}$\text{Q}_{\bold{b}} \big( {\bold{z}} \big)$}%
}}}}
\put(526,-661){\makebox(0,0)[lb]{\smash{{\SetFigFont{8}{9.6}{\familydefault}{\mddefault}{\updefault}{\color[rgb]{0,0,0}$\bold{W}_D^H$}%
}}}}
\put(3376,-211){\makebox(0,0)[lb]{\smash{{\SetFigFont{8}{9.6}{\familydefault}{\mddefault}{\updefault}{\color[rgb]{0,0,0}$\bold{n}$}%
}}}}
\put(2176,389){\makebox(0,0)[lb]{\smash{{\SetFigFont{8}{9.6}{\familydefault}{\mddefault}{\updefault}{\color[rgb]{0,0,0}$\bold{\tilde{x}}$}%
}}}}
\put(151,389){\makebox(0,0)[lb]{\smash{{\SetFigFont{8}{9.6}{\familydefault}{\mddefault}{\updefault}{\color[rgb]{0,0,0}$\bold{x}$}%
}}}}
\put(2626,239){\makebox(0,0)[lb]{\smash{{\SetFigFont{8}{9.6}{\familydefault}{\mddefault}{\updefault}{\color[rgb]{0,0,0}$\bold{H}$}%
}}}}
\put(1576,239){\makebox(0,0)[lb]{\smash{{\SetFigFont{8}{9.6}{\familydefault}{\mddefault}{\updefault}{\color[rgb]{0,0,0}$\bold{F}_A$}%
}}}}
\put(601,239){\makebox(0,0)[lb]{\smash{{\SetFigFont{8}{9.6}{\familydefault}{\mddefault}{\updefault}{\color[rgb]{0,0,0}$\bold{F}_D$}%
}}}}
\put(1201,-511){\makebox(0,0)[lb]{\smash{{\SetFigFont{8}{9.6}{\familydefault}{\mddefault}{\updefault}{\color[rgb]{0,0,0}$\bold{\tilde{y}}$}%
}}}}
\put(2176,-511){\makebox(0,0)[lb]{\smash{{\SetFigFont{8}{9.6}{\familydefault}{\mddefault}{\updefault}{\color[rgb]{0,0,0}$\bold{z}$}%
}}}}
\put(151,-511){\makebox(0,0)[lb]{\smash{{\SetFigFont{8}{9.6}{\familydefault}{\mddefault}{\updefault}{\color[rgb]{0,0,0}$\bold{y}$}%
}}}}
\put(3376,239){\makebox(0,0)[lb]{\smash{{\SetFigFont{9}{10.8}{\rmdefault}{\mddefault}{\updefault}{\color[rgb]{0,0,0}+}%
}}}}
{\color[rgb]{0,0,0}\put(3411,288){\circle{136}}
}%
\end{picture}%